\documentclass[article,lineno]{article}

\usepackage{amsmath}
\usepackage{amsthm}
\usepackage[margin=1.5in]{geometry}
\usepackage{caption}
\usepackage[breaklinks=true]{hyperref}

\usepackage{times}
\usepackage{bm}
\usepackage{natbib}
\usepackage{graphicx}
\usepackage[plain,noend]{algorithm2e}

\usepackage[plain,noend]{algorithm2e}

\makeatletter
\renewcommand{\algocf@captiontext}[2]{\quad #1\algocf@typo. \AlCapFnt{}#2} 
\def\@algocf@capt@plain{top}
\renewcommand{\algocf@makecaption}[2]{%
  \addtolength{\hsize}{\algomargin}%
  \sbox\@tempboxa{\algocf@captiontext{#1}{#2}}%
  \ifdim\wd\@tempboxa >\hsize
    \hskip .5\algomargin%
    \parbox[t]{\hsize}{\algocf@captiontext{#1}{#2}}
  \else%
    \global\@minipagefalse%
    \hbox to\hsize{\box\@tempboxa}
  \fi%
  \addtolength{\hsize}{-\algomargin}%
}
\makeatother

\def\T{{ \mathrm{\scriptscriptstyle T} }}

\DeclareMathOperator*{\argmin}{arg\,min}

\newtheorem{definition}{Definition}
\newtheorem{corollary}{Corollary}
\newtheorem{proposition}{Proposition}

\begin{document}

\title{On the marginal likelihood and cross-validation}

\author{E. Fong \& C. C. Holmes \\ ~ \\ Department of Statistics, University of Oxford, OX1 3LB  \\ edwin.fong@stats.ox.ac.uk \& cholmes@stats.ox.ac.uk}

\maketitle
\begin{abstract}
In Bayesian statistics, the marginal likelihood, also known as the  evidence, is used to evaluate model fit as it quantifies the joint probability of the data under the prior. In contrast, non-Bayesian models are typically compared using cross-validation on held-out  data, either through $k$-fold partitioning or leave-$p$-out subsampling. We show that the marginal likelihood is formally equivalent to exhaustive leave-$p$-out cross-validation averaged over all  values of $p$ and all held-out test sets when using the log posterior predictive probability as the scoring rule. Moreover, the log posterior predictive is the only coherent scoring rule under data exchangeability. This offers new insight into the marginal likelihood and cross-validation  and highlights the potential sensitivity of the marginal likelihood to the choice of the prior. We suggest an alternative approach using cumulative cross-validation following a preparatory training phase. Our work has connections to prequential analysis and intrinsic Bayes factors but is motivated through a different course. 
\end{abstract}

\section{Introduction}

Probabilistic model evaluation and selection is an important task in statistics and machine learning, particularly when multiple models are under initial consideration.  In the non-Bayesian literature, models are typically compared using out-of-sample performance criteria such as cross-validation
\citep{Geisser1979,Shao1993,Vehtari2002}, or predictive information \citep{Watanabe2010}. Computing the leave-$p$-out cross-validation score requires $n$-choose-$p$ test set evaluations  for $n$ data points, which in most cases is computationally unviable and hence approximations such as $k$-fold cross-validation are often used instead \citep{Geisser1975}. A survey is provided by \citet{Arlot2010}, and a Bayesian perspective on cross-validation by \citet{Vehtari2012, Gelman2014}.

In Bayesian statistics, the marginal likelihood or model evidence is the natural measure of model fit. For a model $\mathcal{M}$ with likelihood function or sampling distribution $\left\{f_{\theta}(y): \theta \in \Theta \right\}$ parameterized by $\theta$, a prior $\pi(\theta)$, and observations $y_{1:n} \in \mathcal{Y}^n$, the marginal likelihood or the prior predictive is defined as
\begin{equation}
\label{eq:bml}
p_{\mathcal{M}}( y_{1:n})  = \int f_{\theta}(y_{1:n} ) \, d\pi(\theta) \, .
\end{equation}
The marginal likelihood can be used to calculate the posterior probability of the model given the data, $p( {\cal{M}} \mid y_{1:n} ) \propto p_{\mathcal{M}}( y_{1:n})  \, p({\cal{M}})$, as it is the probability of the data being generated under the prior when the model is correctly specified \cite[Chapter~7]{Robert2007}. The ratio of marginal likelihoods between models is known as the Bayes factor that quantifies the prior to posterior odds on observing the data. The marginal likelihood can be difficult to compute if the likelihood is peaked with respect to the prior, although Monte Carlo solutions exist; see \citet{Robert2009} for a survey.  Under vague priors, the marginal likelihood may also be highly sensitive to the prior dispersion even if the posterior is not; a well known example is Lindley's paradox \citep{Lindley1957,OHagan2004,Robert2014}. As a result, its approximations such as the Bayesian information criterion \citep{Schwarz1978} or the deviance information criterion \citep{Spiegelhalter2002} are widely used, see also \citet{Gelman2014}.

For our work, it is useful to note from the property of probability distributions that the log marginal likelihood can be written as the sum of log conditionals, 
\begin{equation}\label{eq:p_fac}
\log p_{\mathcal{M}}( y_{1:n}) = \sum_{i=1}^n \log p_{\mathcal{M}}(y_i \mid y_{1:i-1}) 
\end{equation}
where $ p_{\mathcal{M}}(y_i \mid y_{1:i-1})  = \int f_{\theta}(y_i) \, d \pi(\theta \mid y_{1:i-1}) $ is the posterior predictive for $i>1$,  ${p_{\cal{M}}(y_1 \mid y_{1:0})} \\ = \int f_{\theta}(y_{1} ) \, d\pi(\theta) \, $,
and this representation is true for any permutation of the data indices. 

While Bayesian inference formally assumes that the model space captures the truth, in the model misspecified or so called $M$-open scenario \cite[Chapter~6]{Bernardo2009} the log marginal likelihood can be simply interpreted as a predictive sequential, or prequential \citep{Dawid1984}, scoring rule of the form $S(y_{1:n}) = \sum_i s(y_i \mid y_{1:i-1}) $ with score function $ s(y_i \mid y_{1:i-1})  = \log {p_{\cal{M}}(y_i \mid y_{1:i-1})}$. This interpretation of the log marginal likelihood as a predictive score \cite[][Chapter~6]{Kass1995,Gneiting2007,Bernardo2009} has resulted in alternative scoring functions for Bayesian model selection \citep{Dawid2014,Dawid2015,Watson2016,Shao2019}, and provides insight into the relationship between the marginal likelihood and posterior predictive methods \citep{Vehtari2012}. \citet{Key1999} considered cross-validation from an $M$-open perspective and introduced a mixture utility for model selection that trades off fidelity to data with predictive power.

\section{Uniqueness of the marginal likelihood under coherent scoring}
 
To begin, we prove that under an assumption of data exchangeability, the log posterior predictive is the only prequential scoring rule that guarantees coherent model evaluation. The coherence property under exchangeability, where the indices of the data points carry no information,  refers to the principle that identical models on seeing the same data should be scored equally irrespective of data ordering.

In demonstrating the uniqueness of the log posterior predictive, it is useful to introduce the notion of a general Bayesian model \citep{Bissiri2016}, which is a framework for Bayesian updating without the requirement of a true model. Define a parameter of interest by
\begin{equation} \label{eq:parameter}
\theta_0 = \argmin_\theta \int l(\theta, y) dF_0(y)
\end{equation}
where $F_0(y)$ is the unknown true sampling distribution giving rise to the data, and  $l : \Theta \times \mathcal{Y} \rightarrow [0,\infty)$ is a loss function linking  an observation $y$ to the parameter $\theta$. \citet{Bissiri2016} argue that after observing $y_{1:n}$, a coherent update of beliefs about $\theta_0$ from a prior $\pi_G(\theta)$  to the posterior $\pi_G(\theta \mid y_{1:n})$ exists and must take on the form 
\begin{equation}\label{eq:genBayespost}
\pi_G(\theta \mid y_{1:n}) \propto \exp\left\{-w l(\theta, y_{1:n}) \right\} \pi_G(\theta)
\end{equation}
where $l(\theta, y_{1:n})= \sum_i l(\theta,y_i)$ is an additive loss function and $w>0$ is a loss scale parameter; see \citet{Holmes2017,Lyddon2019} on the selection of $w$. For $w= 1$ and $l(\theta,y) = -\log f_{\theta}(y)$, we obtain traditional Bayesian updating without assuming the model $f_{\theta}(y)$ is true for some value of $\theta$. From (\ref{eq:parameter}), $M$-open Bayesian inference is simply targeting the value of $\theta$ that minimizes the Kullback-Leibler divergence between $d F_0(y)$ and $f_{\theta}(y)$. The form (\ref{eq:genBayespost}) is uniquely implied by the assumptions in Theorem 1 of \citet{Bissiri2016}, and we now focus on the coherence property of the update rule. An update function $\psi\{l(\theta, y), \pi_G(\theta)\} = \pi_G(\theta \mid y)$ is coherent if, for some inputs $y_{1:2}$, it satisfies 
\begin{equation*} 
\psi[l(\theta, y_2),\psi \{l(\theta, y_1), \pi_G(\theta)\}] =\psi\{l(\theta, y_1)+l(\theta, y_2), \pi_G(\theta)\}.
\end{equation*}
This coherence condition is natural under an assumption of exchangeability as we expect posterior inferences about $\theta_0$ to be unchanged whether we observe $y_{1:2}$ in any order or all at once, as it is in traditional Bayesian updating.

We now extend this coherence condition to general Bayesian model choice, where the goal is to evaluate the fit of the observed data under the general Bayesian model class $\mathcal{M}_G = \{l(\theta,y):{\theta \in \Theta\}}$ with a prior $\pi_G(\theta)$. We treat $w$ as a parameter outside of the model specification, as there are principled methods to select it from the model, prior and data.  We define the log posterior predictive score as
\begin{equation*}
s_G (\tilde{y} \mid y_{1:n}) = \log \int g\{ l(\theta,\tilde{y})\} d\pi_G(\theta \mid y_{1:n})
\end{equation*}
where $g: [0,\infty) \to [0,\infty)$ is a continuous monotonically decreasing scoring function that transforms $l(\theta,y)$ into a predictive score for a test point $\tilde{y}$. We define the cumulative prequential log score as
 \begin{equation*}
S_G(y_{1:n}) = \sum_{i=1}^n s_G (y_i\mid y_{1:i-1}) 
\end{equation*}
where $s_G ( y_1 \mid y_{1:0})= \log \int g\{ l(\theta,y_1)\} d\pi_G(\theta ) $. The cumulative prequential log score sums the log posterior predictive score of each consecutive data point in a prequential manner, where a large score indicates that the model is predicting well. An intuitive choice for the scoring function might be the negative loss  $g(l) = -l$, but we will see that this violates coherency, as defined below.
\begin{definition}
The model scoring function $g(l)$ is coherent if it satisfies 
\begin{equation} \label{eq:coherentscore}
\sum_{i=1}^n s_G (y_i\mid y_{1:i-1})  = \log\int g\{l(\theta,y_{1:n})\} d\pi_G(\theta)
\end{equation}
for all $\Theta$, $\pi(\theta)$ and $n>0$, such that $S_G(y_{1:n})$ is invariant to the ordering or partitioning of the observations.
\end{definition}
We now present our main result on the uniqueness of the choice of $g$.

\begin{proposition}\label{prop1}
If the model scoring function  $g: [0,\infty) \to [0,\infty)$ is continuous, monotonically decreasing and coherent, then the unique choice of scoring rule $g(l)$  is
\begin{equation*}
g(l) = \exp(-wl)
\end{equation*}
where $w$ is the loss-scale in the general Bayesian posterior.
\end{proposition}
\begin{proof}
The proof is given in the Supplementary Material.
\end{proof}
This holds irrespective of whether the model is true or not. More importantly for us is the corollary below.

\begin{corollary}
The marginal likelihood is the unique coherent marginal score for Bayesian inference.
\end{corollary}
\begin{proof}
Let $w=1$ and $l(\theta,y) = -\log f_{\theta}(y)$, and hence $g\{ l(\theta,y)\} = f_{\theta}(y)$. 
\end{proof}
The marginal likelihood arises naturally as the unique prequential scoring rule under coherent belief updating in the Bayesian framework. The coherence of the marginal likelihood implies an invariance to the permutation of the observations $y_{1:n}$ under exchangeability, including independent and identically distributed data, a property that is not shared by other prequential scoring rules, such as \citet{Dawid2014, Grunwald2017, Shao2019}.

\section{The marginal likelihood and cross-validation}

\subsection{Equivalence of the marginal likelihood and cumulative \newline  cross-validation}
The leave-$p$-out cross-validation score is defined as 
\begin{equation} \label{eq:SCV}
S_{CV} (y_{1:n} ; p)  =   \frac{1}{{n \choose p}}  \sum_{t=1}^{{n \choose p }} \frac{1}{p}  \sum_{j=1}^{p} s\left(\tilde{y}_{j}^{(t)} \mid y^{(t)}_{1:n-p}\right) 
\end{equation}
 where $\tilde{y}^{(t)}_{1:p}$ denotes the $t$th of $n$-choose-$p$ possible held-out test sets,  with $y_{1:n-p}^{(t)}$ the corresponding training set, such that $y_{1:n} = \left\{\tilde{y}^{(t)}, y^{(t)}\right\}$, and $S_{CV}$ records the average predictive score per datum. Although leave-one-out cross-validation is a popular choice, it was shown in \citet{Shao1993} that it is asymptotically inconsistent for a linear model selection problem, and requires $\left(p/n \right) \to 1$ as $n \to \infty$ for consistency. We will not go into further detail here but instead refer the reader to \cite{Arlot2010}. Selecting a larger $p$ has the interpretation of penalizing complexity \citep{Vehtari2012}, as complex models will tend to over-fit to a small training set. However, the number of test set evaluations grows rapidly with $p$ and hence $k$-fold cross-validation is often adopted for computational convenience. 
 
From a Bayesian perspective it is natural to consider the log posterior predictive as the scoring function, $s(\tilde{y} \mid y) = \log \int f_\theta(\tilde{y}) d\pi(\theta \mid  y)$, particularly as we have now shown that it is the only coherent scoring mechanism, which leads us to the following result.

\begin{proposition}\label{prop2}
The Bayesian marginal likelihood is equivalent to the cumulative leave-$p$-out cross-validation score using the log posterior predictive as the scoring rule, such that 
\begin{equation} \label{eq:margcv}
 \log p_{\cal{M}}(y_{1:n}) = \sum_{p=1}^{n}  S_{CV} (y_{1:n} ; p) 
\end{equation}
with  $s(\tilde{y}_j \mid y_{1:n-p}) = \log p_{\cal{M}}(\tilde{y}_j \mid y_{1:n-p}) = \log \int f_\theta(\tilde{y}_j) \, d\pi(\theta \mid  y_{1:n-p})$. 
\end{proposition}

\begin{proof} 
This follows from the invariance of the marginal likelihood under arbitrary permutation of the sequence $y_{1:n}$ in (\ref{eq:p_fac}). We provide a proof and an alternative proof by induction in the Supplementary Material.
\end{proof}

The Bayesian marginal likelihood is simply $n$ times the average leave-$p$-out cross-validation score, $n \times (1/n) \sum_{p=1}^{n}  S_{CV} (y_{1:n} ; p)$,  where the scaling by $n$ is due to  (\ref{eq:SCV}) being a per datum score. Bayesian models are  evaluated through out-of-sample predictions on all  $(2^n-1)$ possible held-out test sets whereas cross-validation with fixed $p$ only captures a snapshot of model performance. Evaluating the predictive performance on $(2^n-1)$ test sets would appear intractable for most applications, but we see through (\ref{eq:margcv}) and (\ref{eq:bml}) that it is computable as a single integral.

\subsection{Sensitivity to the prior and preparatory training}

The representation of the marginal likelihood as a cumulative cross-validation score (\ref{eq:margcv}) provides insight into the sensitivity to the prior. The last term in the right hand side of (\ref{eq:margcv}) involves no training data, $S_{CV}(y_{1:n}; n) = (1/n)\sum_{i=1}^n \log \int f_{\theta}(y_i) \, d\pi(\theta)$, which scores the model entirely on how well the analyst is able to specify the prior. In many situations, the analyst may not want this term to contribute to model evaluation. Moreover, there is tension between any desire to specify vague priors to safeguard their influence and the fact that diffuse priors can lead to an arbitrarily large and negative model score for real valued parameters from (\ref{eq:margcv}).  It may seem inappropriate to penalize a model based on the subjective ability to specify the prior, or to compare models using a score that includes contributions from predictions made using only a handful of training points even with informative priors. For example, we see that 10\% of terms contributing to the marginal likelihood come from out-of-sample predictions using, on average, less than 5\% of available training data. 
This is related to the start-up problem in prequential analysis \citep{Dawid1992}.
 
A natural and obvious solution is to begin evaluating the model performance after a preparatory phase, for example using 10\% or 50\% of the data as preparatory training prior to testing. This leads to a Bayesian cumulative leave-$P$-out cross-validation score defined as  
\begin{equation}\label{eq:CCV}
S_{CCV}(y_{1:n}; P) = \sum_{p=1}^{P} S_{CV} (y_{1:n} ; p) 
\end{equation}
with a  preparatory cross-validation score 
$
S_{PCV}(y_{1:n}; P)  =   \sum_{p=P+1}^{n} S_{CV} (y_{1:n} ; p),  
$
for $1 \leq P < n$. We suggest setting $P$ to leave out $0.9n$, $0.5n$ or $\max(0.9n, n-10d)$, where $d$ is the total number of model parameters, as reasonable default choices, but clearly this is situation specific. One may be interested in reporting both $S_{CCV}$ and $S_{PCV}$, as the latter can be regarded as an evaluation of the prior, but we suggest that only $S_{CCV}$ is used for model evaluation from the arguments above. Although full coherency is now lost, we still have coherency conditioned on a preparatory training set, where permutation of the data within the training and test sets does not affect the score, and so we can write (\ref{eq:CCV}) as
\begin{equation}\label{eq:CCV2}
S_{CCV}(y_{1:n}; P)  = \frac{1}{{n \choose P }} \sum_{t=1}^{n \choose P } \log p_\mathcal{M}\left(\tilde{y}^{(t)}_{1:P} \mid y^{(t)}_{1:n-P}\right).
\end{equation}
This equivalence is derived in the Supplementary Material in a similar fashion to Proposition \ref{prop2}. This has precisely the form of the the log geometric intrinsic Bayes factor of \citet{Berger1996} but motivated by a different route.  The intrinsic Bayes factor was developed in an objective Bayesian setting  \citep{Berger2001}, where improper priors cause indeterminacies in the evaluation of the marginal likelihood. The intrinsic Bayes factor remedies this with a partition of the data into $y_{1:l},y_{l+1:n}$, where $y_{1:l}$ is the minimum training sample used to convert an improper prior $\pi(\theta)$ into a proper prior  $\pi(\theta \mid y_{1:l})$. In contrast, we set $n-P$ to provide preparatory training and  $\pi(\theta)$ can be subjective. Moreover, in modern applications we often have $d \gg n$ where intrinsic Bayes factors cannot be applied in their original form.
\newpage
We can approximate (\ref{eq:CCV2}) through Monte Carlo where the training data sets ${y}^{(t)}_{1:n-P}$ are drawn uniformly at random, and for non-conjugate models the inner term must also be estimated, for example through
\begin{equation}\label{eq:MCCCV}
\hat{S}_{CCV}(y_{1:n}; P) = \frac{1}{T} \sum_{t=1}^{T} \log \left\{ \frac{1}{B}\sum_{b=1}^B f_{\theta_b^{(t)}}\left(\tilde{y}^{(t)}_{1:P}\right)\right\}
\end{equation}
where samples $\theta_b^{(t)} \sim \pi \left(\theta \mid y^{(t)}_{1:n-P}\right)$  are obtained via $T$ Markov chain Monte Carlo samplers. If we assume that the number of samples $B$ per chain is sufficiently large, then the variance of the estimate $\hat{S}_{CCV}$ is approximately of the form $\tau^2 / T$. However, fitting $T$ models may be costly, but we can run the chains in parallel. To avoid the need for $T$ Markov chain Monte Carlo chains in (\ref{eq:MCCCV}), we can instead take advantage of the fact that the partial posteriors for different training sets will be similar, and utilize importance sampling \citep{Bhattacharya2007, Vehtari2017} or sequential Monte Carlo \citep{Bornn2010} to estimate the posterior predictives for computational savings. We provide further details on efficient computation of (\ref{eq:MCCCV}) in the Supplementary Material.  

\section{Illustration for the normal linear model}

We illustrate the use of Bayesian cumulative cross-validation in a polynomial regression example, where the $r$th polynomial model is defined as
\begin{equation*}
f_{\theta}(y \mid x,r) =\mathcal{N}\{y;  \theta^\T \phi_r(x), \sigma^2 \}, \quad \phi_r(x) = \begin{bmatrix} 1 & x &\ldots & x^{r-1} &x^r \end{bmatrix}^\T.
\end{equation*}
We observe the data $\{y_{1:n},x_{1:n}\}$, and we place a fixed vague prior on the intercept term, $\theta_0 \sim \mathcal{N}(\theta_0; 0,100^2)$,  and $\theta_d \sim \mathcal{N}(\theta_d; 0,s^2)$ for $d \in \{1,\ldots,r\}$ on the remaining coefficients. 
In our example, we have $n=100$ and the true model is $r=1$, $\theta = \begin{bmatrix} 1 & 0.5\end{bmatrix}^\T$ with known $\sigma^2 = 1$. For our prior, we vary the value of $s^2 \in \left\{10^{-1},10^0,10^4 \right\}$ to investigate the impact of the prior tails. For each prior setting, we calculate  $\log p_\mathcal{M}(y_{1:n})$ and $S_{CCV}(y_{1:n};P)$ for  models $r \in \{0,1,2\}$. In this example, $\log p_\mathcal{M}(y_{1:n})$  is tractable, whereas $S_{CCV}$ requires a Monte Carlo average over tractable log posterior predictives. We report the mean over 10 runs of estimating $S_{CCV}$ with $T= 10^6$ random training/test splits. We calculate the Monte Carlo standard error over the 10 runs and report the maximum for each setting of $P$.

The results are shown in Table \ref{tbl:normal}, where $\hat{S}_{CCV}$ is normalized to the same scale as $\log p_r(y_{1:n})$. Under the strong prior $s^2 = 10^{-1}$ and the moderate prior $s^2 = 10^0$, the marginal likelihood correctly identifies the true model, but when we increase $s^2$ to $10^{4}$ it heavily over-penalizes the more complex models and prefers $r=0$. In fact, the magnitude of the marginal likelihood and the discrepancy just described can be made arbitrarily large by simply increasing $s^2$, which should be guarded against when a modeller has weak prior beliefs. This issue is not observed with $\hat{S}_{CCV}$ for the values of $P$ we consider. The vague prior does not impede the ability of $\hat{S}_{CCV}$ to correctly identify the true model $r=1$ and the scores are stable within each column of $P$. 

In the Supplementary Material, we present graphical tools for exploring the cumulative cross-validation and the effect of the choice of $P$ on $S_{CCV}$. We provide an additional example using probit regression on the Pima Indian data set.

\begin{table}[!h]
\center
\def~{\hphantom{-}}
\caption{ Log marginal likelihoods and cumulative cross-validation scores for normal linear model}{%
\begin{small}
\begin{sc}
\begin{tabular}{l|c||c|c|c|c}
$s^2$~~&Model  &$\log p_r(y_{1:n})$& \multicolumn{3}{c} {$\hat{S}_{CCV}(y_{1:n}; P) \times n/P$} \\[2pt]
 & $r$&  &$P=0.9n$& $P= 0.5n$ & $P = 0.1n$ \\[5pt]
 \hline\hline $10$\rlap{$^{-1}$} 		  &0& -158.82 &-153.80&-153.21& -153.06  \\
  							  &1& {-155.57} & {-150.39}&{-149.55}&{-149.27}\\
  							  &2& -156.12 &-150.94&-149.81& -149.38  \\[5pt]
$10$\rlap{$^{0}$}   		  &0& -158.82 &-153.80&-153.21& -153.06 \\
 							  &1& {-156.26} &{-150.77}&{-149.66}& {-149.34} \\
							  &2& -157.80&-151.90&-150.04& -149.50 \\[5pt]
$10$\rlap{$^{4}$}   		  &0& {-158.82 }&-153.80&-153.21& -153.06 \\
							  &1& -160.81 &{-150.91}&{-149.68}& {-149.35}  \\
 							  &2&  -166.93 &-152.30&-150.08&  -149.53\\[5pt]
\hline \hline \multicolumn{3}{r|} {Maximum standard error} &\hphantom{-000}0.002&\hphantom{-000}0.008& \hphantom{-000}0.023
\end{tabular}
\end{sc}
\end{small}}
\label{tbl:normal}
\end{table}

\newpage
\section{Discussion}

We have shown that for coherence, the unique scoring rule for Bayesian model evaluation in either $M$-open or $M$-closed is provided by the log posterior predictive probability, and that the marginal likelihood is equivalent to a cumulative cross-validation score over all training-test data partitions. The coherence flows from the fact that the scoring rule and the Bayesian update both use the same information, namely the likelihood function, which is appropriate as the alternative would be to learn and score under different criteria. If we are interested in an alternative loss function to the log likelihood, we advocate a general Bayesian update \citep{Bissiri2016,Lyddon2019} that targets the parameters minimising the expected loss, with models evaluated using the corresponding coherent cumulative cross-validation score. 

\section*{Acknowledgement}
The authors thank Lucian Chan, George Nicholson, the editor, an associate editor and two referees for their helpful comments. Fong was funded by The Alan Turing Institute. Holmes was supported by The Alan Turing Institute, the Health Data Research, U.K., the Li Ka Shing Foundation, the Medical Research Council, and the U.K. Engineering and Physical Sciences Research Council.

\bibliographystyle{apalike}
\bibliography{paper-ref}

\begin{thebibliography}{}

\bibitem[Arlot and Celisse, 2010]{Arlot2010}
Arlot, S. and Celisse, A. (2010).
\newblock A survey of cross-validation procedures for model selection.
\newblock {\em Statistics Surveys}, 4:40--79.

\bibitem[Berger and Pericchi, 1996]{Berger1996}
Berger, J.~O. and Pericchi, L.~R. (1996).
\newblock The intrinsic {B}ayes factor for model selection and prediction.
\newblock {\em Journal of the American Statistical Association},
  91(433):109--122.

\bibitem[Berger and Pericchi, 2001]{Berger2001}
Berger, J.~O. and Pericchi, L.~R. (2001).
\newblock {\em Objective Bayesian Methods for Model Selection: Introduction and
  Comparison}, volume~38 of {\em Lecture Notes--Monograph Series}, pages
  135--207.
\newblock Institute of Mathematical Statistics, Beachwood, OH.

\bibitem[Bernardo and Smith, 2009]{Bernardo2009}
Bernardo, J. and Smith, A. (2009).
\newblock {\em Bayesian Theory}.
\newblock Wiley Series in Probability and Statistics. Wiley.

\bibitem[Bhattacharya and Haslett, 2007]{Bhattacharya2007}
Bhattacharya, S. and Haslett, J. (2007).
\newblock Importance re-sampling {MCMC} for cross-validation in inverse
  problems.
\newblock {\em Bayesian Analysis}, 2(2):385--407.

\bibitem[Bissiri et~al., 2016]{Bissiri2016}
Bissiri, P.~G., Holmes, C.~C., and Walker, S.~G. (2016).
\newblock {A general framework for updating belief distributions}.
\newblock {\em Journal of the Royal Statistical Society: Series B (Statistical
  Methodology)}, 78(5):1103--1130.

\bibitem[Bornn et~al., 2010]{Bornn2010}
Bornn, L., Doucet, A., and Gottardo, R. (2010).
\newblock An efficient computational approach for prior sensitivity analysis
  and cross-validation.
\newblock {\em Canadian Journal of Statistics}, 38(1):47--64.

\bibitem[Dawid, 1984]{Dawid1984}
Dawid, A.~P. (1984).
\newblock {Present Position and Potential Developments: Some Personal Views:
  Statistical Theory: The Prequential Approach}.
\newblock {\em Journal of the Royal Statistical Society. Series A (General)},
  147(2):278.

\bibitem[Dawid, 1992]{Dawid1992}
Dawid, A.~P. (1992).
\newblock Prequential analysis, stochastic complexity and {B}ayesian inference.
\newblock {\em Bayesian Statistics}, 4:109--125.

\bibitem[Dawid and Musio, 2014]{Dawid2014}
Dawid, A.~P. and Musio, M. (2014).
\newblock Theory and applications of proper scoring rules.
\newblock {\em METRON}, 72(2):169--183.

\bibitem[Dawid and Musio, 2015]{Dawid2015}
Dawid, A.~P. and Musio, M. (2015).
\newblock {B}ayesian model selection based on proper scoring rules.
\newblock {\em Bayesian Analysis}, 10(2):479--499.

\bibitem[Geisser, 1975]{Geisser1975}
Geisser, S. (1975).
\newblock The predictive sample reuse method with applications.
\newblock {\em Journal of the American Statistical Association},
  70(350):320--328.

\bibitem[Geisser and Eddy, 1979]{Geisser1979}
Geisser, S. and Eddy, W. (1979).
\newblock A predictive approach to model selection.
\newblock {\em Journal of the American Statistical Association}, 74:153--160.

\bibitem[Gelman et~al., 2014]{Gelman2014}
Gelman, A., Hwang, J., and Vehtari, A. (2014).
\newblock {Understanding predictive information criteria for Bayesian models}.
\newblock {\em Statistics and Computing}, 24:997--1016.

\bibitem[Gneiting and Raftery, 2007]{Gneiting2007}
Gneiting, T. and Raftery, A.~E. (2007).
\newblock Strictly proper scoring rules, prediction, and estimation.
\newblock {\em Journal of the American Statistical Association},
  102(477):359--378.

\bibitem[Gr{\"u}nwald and van Ommen, 2017]{Grunwald2017}
Gr{\"u}nwald, P. and van Ommen, T. (2017).
\newblock Inconsistency of {B}ayesian inference for misspecified linear models,
  and a proposal for repairing it.
\newblock {\em Bayesian Analysis}, 12(4):1069--1103.

\bibitem[Holmes and Walker, 2017]{Holmes2017}
Holmes, C.~C. and Walker, S.~G. (2017).
\newblock {Assigning a value to a power likelihood in a general Bayesian
  model}.
\newblock {\em Biometrika}, 104(2):497--503.

\bibitem[Kass and Raftery, 1995]{Kass1995}
Kass, R.~E. and Raftery, A.~E. (1995).
\newblock Bayes factors.
\newblock {\em Journal of the American Statistical Association},
  90(430):773--795.

\bibitem[Key et~al., 1999]{Key1999}
Key, J., Pericchi, L., and Smith, A. (1999).
\newblock {B}ayesian model choice: What and why? (with discussion).
\newblock {\em Bayesian Statistics}, 5.

\bibitem[Lindley, 1957]{Lindley1957}
Lindley, D.~V. (1957).
\newblock {A statistical paradox}.
\newblock {\em Biometrika}, 44(1-2):187--192.

\bibitem[Lyddon et~al., 2019]{Lyddon2019}
Lyddon, S.~P., Holmes, C.~C., and Walker, S.~G. (2019).
\newblock {General Bayesian updating and the loss-likelihood bootstrap}.
\newblock {\em Biometrika}, 106(2):465--478.

\bibitem[Marin and Robert, 2010]{Marin2010}
Marin, J.-M. and Robert, C.~P. (2010).
\newblock {Importance sampling methods for Bayesian discrimination between
  embedded models}.
\newblock {\em Frontiers of Statistical Decision Making and Bayesian Analysis},
  pages 513--527.

\bibitem[O'Hagan and Forster, 2004]{OHagan2004}
O'Hagan, A. and Forster, J.~J. (2004).
\newblock {\em Kendall's Advanced Theory of Statistics, volume 2B: Bayesian
  Inference}, volume~2.
\newblock Arnold.

\bibitem[Rischard et~al., 2018]{Rischard2018}
Rischard, M., Jacob, P.~E., and Pillai, N. (2018).
\newblock Unbiased estimation of log normalizing constants with applications to
  {B}ayesian cross-validation.
\newblock {\em arXiv preprint arXiv:1810.01382}.

\bibitem[Robert, 2007]{Robert2007}
Robert, C.~P. (2007).
\newblock {\em {The Bayesian Choice: From Decision-Theoretic Foundations to
  Computational Implementation}}.
\newblock Springer, 2nd edition.

\bibitem[Robert, 2014]{Robert2014}
Robert, C.~P. (2014).
\newblock {On the Jeffreys-Lindley paradox}.
\newblock {\em Philosophy of Science}, 81(2):216--232.

\bibitem[Robert and Wraith, 2009]{Robert2009}
Robert, C.~P. and Wraith, D. (2009).
\newblock Computational methods for {B}ayesian model choice.
\newblock In {\em AIP conference proceedings}, volume 1193, pages 251--262.
  AIP.

\bibitem[Schwarz, 1978]{Schwarz1978}
Schwarz, G. (1978).
\newblock Estimating the dimension of a model.
\newblock {\em The Annals of Statistics}, 6(2):461--464.

\bibitem[Shao, 1993]{Shao1993}
Shao, J. (1993).
\newblock Linear model selection by cross-validation.
\newblock {\em Journal of the American Statistical Association},
  88(422):486--494.

\bibitem[Shao et~al., 2019]{Shao2019}
Shao, S., Jacob, P.~E., Ding, J., and Tarokh, V. (2019).
\newblock Bayesian model comparison with the {H}yv{\"a}rinen score: Computation
  and consistency.
\newblock {\em Journal of the American Statistical Association}, pages 1--24.

\bibitem[Spiegelhalter et~al., 2002]{Spiegelhalter2002}
Spiegelhalter, D.~J., Best, N.~G., Carlin, B.~P., and Van Der~Linde, A. (2002).
\newblock Bayesian measures of model complexity and fit.
\newblock {\em Journal of the Royal Statistical Society: Series B (Statistical
  Methodology)}, 64(4):583--639.

\bibitem[Vehtari et~al., 2017]{Vehtari2017}
Vehtari, A., Gelman, A., and Gabry, J. (2017).
\newblock Practical {B}ayesian model evaluation using leave-one-out
  cross-validation and {WAIC}.
\newblock {\em Statistics and Computing}, 27(5):1413--1432.

\bibitem[Vehtari and Lampinen, 2002]{Vehtari2002}
Vehtari, A. and Lampinen, J. (2002).
\newblock Bayesian model assessment and comparison using cross-validation
  predictive densities.
\newblock {\em Neural Computation}, 14(10):2339--2468.

\bibitem[Vehtari and Ojanen, 2012]{Vehtari2012}
Vehtari, A. and Ojanen, J. (2012).
\newblock A survey of {B}ayesian predictive methods for model assessment,
  selection and comparison.
\newblock {\em Statistics Surveys}, 6:142--228.

\bibitem[Watanabe, 2010]{Watanabe2010}
Watanabe, S. (2010).
\newblock Asymptotic equivalence of {B}ayes cross validation and widely
  applicable information criterion in singular learning theory.
\newblock {\em Journal of Machine Learning Research}, 11:3571--3594.

\bibitem[Watson and Holmes, 2016]{Watson2016}
Watson, J. and Holmes, C.~C. (2016).
\newblock Approximate models and robust decisions.
\newblock {\em Statistical Science}, 31(4):465--489.

\end{thebibliography}
\appendix
\numberwithin{equation}{subsection}
\numberwithin{figure}{subsection}
\numberwithin{table}{subsection}
\setcounter{figure}{0}
\newpage

\newpage
\section{Supplementary Material}

\subsection{Proof of Proposition 1}\label{proof1}
\begin{proof}
We look at the case where $\Theta = \left\{ 0,1\right\}$, so the prior $\pi_G(\theta)$ is parametrized by  $p \in [0,1]$ with $\pi_G(\theta = 0)=p$. We let $n=2$,  denoting the observables as $y_1,y_2$.  We further denote $l(0,y_1) = l_0$ and $l(1,y_1) = l_1$, and likewise $l(0,y_2) = h_0$ and $l(1,y_2) = h_1$. We write $p_1$ as the updated $\pi_G(\theta = 0 \mid y_1)$ obtained from the general Bayesian update (\ref{eq:genBayespost}).  The function $g(l)$ must then satisfy
\begin{equation}\label{eq:A1}
\begin{aligned}
\left\{g(l_0)p + g(l_1)(1-p) \right\}\left\{ g(h_0)p_1 + g(h_1)(1-p_1) \right\} \\ = \left\{ g(l_0+h_0)p + g(l_1 + h_1)(1-p) \right\}
\end{aligned}
\end{equation}
for all $0\leq p \leq 1$ and for all $l_0,l_1,h_0,h_1 \in [0,\infty)$. If we let $p =1$, then $p_1 = 1$, so this simplifies to
\begin{equation*}
g(l_0)g(h_0) = g(l_0 + h_0).
\end{equation*}
As $g$ is continuous and monotonically decreasing, to satisfy (\ref{eq:coherentscore}) it must take on the form
\begin{equation}\label{eq:A2}
g(l) = \exp(-\lambda l)
\end{equation}
for $\lambda \geq 0$.  
We now explicitly write out the form of $p_1$
\begin{equation}\label{eq:A3}
p_1 = \frac{\exp (-wl_0)p}{\exp (-wl_0)p + \exp (-wl_1)(1-p)} =  \frac{\exp (-wl_0)p}{Z_1}.
\end{equation}
If we plug (\ref{eq:A2}), (\ref{eq:A3}) into (\ref{eq:A1}), we obtain
\begin{equation*}
\begin{aligned}
\left\{\exp(-\lambda l_0)p+ \exp(-\lambda l_1)(1-p) \right\} \left\{ \exp(-\lambda h_0) \exp (-wl_0)p + \exp(-\lambda h_1) \exp (-wl_1)(1-p)\right\} \\= Z_1 \left[ \exp\{-\lambda (l_0+h_0)\}p + \exp\{-\lambda (l_1 + h_1)\} (1-p)\right].
\end{aligned}
\end{equation*}
Expanding, cancelling terms, and simplifying we obtain 
\begin{equation*}
\begin{aligned}
\exp(-\lambda l_1  - wl_0) \{\exp(-\lambda h_0) -\exp(-\lambda h_1)\}\\ = \exp(-\lambda l_0  - wl_1) \{\exp(-\lambda h_0) -\exp(-\lambda h_1)\}
\end{aligned}
\end{equation*}
and so we must have $\lambda = 0$ or $\lambda = w$, where only the latter solution is non-trivial. We have thus shown that for $n=2, |\Theta| = 2$, the unique non-trivial solution to (\ref{eq:coherentscore}) is
\begin{equation} \label{eq:A4}
g(l) = \exp(-w l).
\end{equation}
The remainder of the proof involves showing that this choice of $g$ satisfies (\ref{eq:coherentscore}) for all $n>0$ and all $\Theta$ and $\pi(\theta)$. Subbing (\ref{eq:A4}) into (\ref{eq:coherentscore}), we obtain
\begin{equation*}
\begin{aligned}
\prod_{i=1}^n \exp \left\{ s_G (y_i\mid y_{1:i-1}) \right\}&=   \prod_{i=1}^n  \int \exp\{-w l(\theta,y_i)\} \frac{\exp \left\{-w l(\theta,y_{1:i-1})\right\} d\pi_G(\theta)}{\int \exp \left\{-w l(\theta',y_{1:i-1}) \right\} d\pi_G(\theta')}   \\ 
&= \prod_{i=1}^n \frac{\int \exp \left\{-w l(\theta,y_{1:i}) \right\} d\pi_G(\theta)}{\int \exp \left\{-w l(\theta',y_{1:i-1}) \right\} d\pi_G(\theta')} \, \,
\\&=  \int \exp \left\{-w l(\theta,y_{1:n}) \right\} d\pi_G(\theta) 
\end{aligned} 
\end{equation*}
where for convenience we denote $l(\theta,y_{1:0})=0$. 
\end{proof}

\subsection{Proof of Proposition 2}
\begin{proof}
Consider the $(n! \times n)$ matrix $Z$ with elements $(Z)_{ti} = \log p_{\cal{M}}(y^{(t)}_{i} \mid y^{(t)}_{1:i-1})$, such that the $t$th row of $Z$ records the prequential sequence of log posterior predictives under the $t$th of $n!$ permutations of $y_{1:n}$. By the property of conditional probabilities, we have that the row sums of $Z$ are equal, $\sum_i (Z)_{ti} = \sum_i (Z)_{t'i}$ for all $t, t'$, and hence 
\begin{equation*} 
\label{eq:nfac}
 \begin{aligned}
\log  p_{\cal{M}}(y_{1:n})   = \frac{1}{n!} \sum_{t = 1}^{n!} \sum_{i=1}^n   (Z)_{ti}   =\sum_{i=1}^n   \frac{1}{n!} \sum_{t = 1}^{n!} (Z)_{ti}. 
 \end{aligned}
\end{equation*}

Within each column of $Z$, the values $(Z)_{ti}$ are invariant to the permutation of $y_{1:i-1}$ in the preceding $i-1$ columns under exchangeability. There are thus $n$-choose-$(i-1)$ distinct training sets and $n-i+1$ choices for $y_i$ given the training set. For each column $i \in \{1, \ldots, n\}$, we can then write
\begin{equation*} 
\begin{aligned}
\frac{1}{n!} \sum_{t = 1}^{n!} (Z)_{ti} &=  \frac{1}{{n \choose i-1}} \sum_{t=1}^{{n \choose i-1}} \frac{1}{n-i+1} \sum_{j=1}^{n-i+1} s\left(\tilde{y} _{j}^{(t)} \mid y_{1:i-1}^{(t)}\right) \\ &= S_{CV} (y_{1:n} ;n- i+1) ~~~~ 
\end{aligned}
\end{equation*}
where $s\left(\tilde{y} _{j}^{(t)} \mid y_{1:i-1}^{(t)}\right) = \log p_{\mathcal{M}}\left(\tilde{y} _{j}^{(t)} \mid y_{1:i-1}^{(t)}\right)$. We have the result for $p =n-i+1$.
\end{proof}

\subsection{Alternative proof of Proposition 2}\label{proof22}
To prove Proposition \ref{prop2}, we first begin by showing the following proposition.
\begin{proposition}\label{prop3}
For a preparatory cross-validation score, $S_{PCV}(y_{1:n}; P)$, defined as the sum of cross-validation terms from leave-$(P+1)$-out to leave-$n$-out,
\begin{equation*}
 S_{PCV}(y_{1:n}; P) =   \sum_{p=P+1}^{n} S_{CV} (y_{1:n} ; p),
\end{equation*}
we have the following equivalence relationship
\begin{equation}\label{eq:PCV_marg}
S_{PCV}(y_{1:n}; P) =  \frac{1}{{n \choose P}}  \sum_{t=1}^{{n \choose P }} \log p_{\mathcal{M}}\left(y_{1:n-P}^{(t)}\right)
\end{equation}
which states that $S_{PCV}$ is the average log marginal likelihood over all choices of the training set.
\end{proposition}
\begin{proof}
 To show this, we use a proof by induction. We see that (\ref{eq:PCV_marg}) is trivially true for $P =n-1$, as this is simply $S_{CV}(y_{1:n};n)$. Assuming (\ref{eq:PCV_marg}) holds for some $1\leq P\leq n-1$, we have
 
\begin{equation*}
\begin{aligned}
S_{PCV}(y_{1:n}; P-1) &= S_{PCV}(y_{1:n}; P) + S_{CV}(y_{1:n};P)\\
&= \frac{1}{{n \choose P}}  \sum_{t=1}^{{n \choose P }} \log p_{\mathcal{M}}\left(y_{1:n-P}^{(t)}\right) +  \frac{1}{{n \choose P}}  \sum_{t=1}^{{n \choose P }} \frac{1}{P}  \sum_{j=1}^{P} \log p_{\mathcal{M}}\left(\tilde{y}_{j}^{(t)} \mid y^{(t)}_{1:n-P}\right) \\
&= \frac{1}{P {n \choose P}}  \sum_{t=1}^{{n \choose P }} \left\{P \log p_{\mathcal{M}}\left(y_{1:n-P}^{(t)}\right) +   \sum_{j=1}^{P} \log p_{\mathcal{M}}\left(\tilde{y}_{j}^{(t)} \mid y^{(t)}_{1:n-P}\right)  \right\}.
\end{aligned}
\end{equation*}
From the properties of conditional probability, we can write 
\begin{equation}\label{eq:SPCVproof}
\begin{aligned}
S_{PCV}(y_{1:n}; P-1) &=  \frac{1}{P {n \choose P}}  \sum_{t=1}^{{n \choose P }}  \sum_{j=1}^{P} \log p_{\mathcal{M}}\left(\tilde{y}_{j}^{(t)}, y^{(t)}_{1:n-P}\right)  .
\end{aligned}
\end{equation}
Again, the marginal likelihood is invariant to the permutation of the sequence under data exchangeability, so we have to consider the repetitions in the partitions $\tilde{y}_{j}^{(t)}, y^{(t)}_{1:n-P}$. For each of the $n$ choose $(n-P+1)$ unordered sequences $y^{(t')}_{1:n-P+1}$, there are $(n-P+1)$ partitions into $\tilde{y}_{j}^{(t)}, y_{1:n-P}^{(t)}$, so there are $n-P+1$ repetitions of each unordered  $y^{(t')}_{1:n-P+1}$ in (\ref{eq:SPCVproof}). We can thus write 
\begin{equation*}
\begin{aligned}
S_{PCV}(y_{1:n}; P-1) &=  \frac{(n-P+1)}{P {n \choose P}}  \sum_{t'=1}^{{n \choose P-1 }}   \log p_{\mathcal{M}}\left(y^{(t')}_{1:n-P+1}\right)  \\ &= \frac{1}{{n \choose P-1}}  \sum_{t'=1}^{{n \choose P-1 }}   \log p_{\mathcal{M}}\left(y^{(t')}_{1:n-P+1}\right) 
\end{aligned}
\end{equation*}
 and by induction we have  (\ref{eq:PCV_marg}). 
\end{proof} 
Proposition 2 then follows trivially by setting $P = 0$ in Proposition \ref{prop3}.

\subsection{Derivation of $S_{CCV}$ for Bayesian models}\label{SCCVderivation}
The following corollary follows easily from Propositions \ref{prop2} and \ref{prop3}.
\begin{corollary}
For the cumulative cross-validation score defined as
\begin{equation}\label{eq:CCV_supp}
S_{CCV}(y_{1:n}; P) = \sum_{p=1}^{P} S_{CV} (y_{1:n} ; p),
\end{equation}
we have the following equivalence relationship
\begin{equation}\label{eq:CCV_marg}
S_{CCV}(y_{1:n}; P)  = \frac{1}{{n \choose P }} \sum_{t=1}^{n \choose P } \log p_\mathcal{M}\left(\tilde{y}^{(t)}_{1:P} \mid y^{(t)}_{1:n-P}\right).
\end{equation}
\end{corollary}
\begin{proof}
We note that $\log p_{\mathcal{M}} (y_{1:n}) = S_{CCV}(y_{1:n}; P) + S_{PCV}(y_{1:n};P)$ from their definitions and Proposition \ref{prop2}. From the permutation invariance of the marginal likelihood, we can write 
\begin{equation} \label{eq:SCCVproof2}
\log p_{\mathcal{M}} (y_{1:n}) = \frac{1}{{n \choose P}} \sum_{t=1}^{n \choose P} \log p_{\mathcal{M}}\left(\tilde{y}^{(t)}_{1:P},y^{(t)}_{1:n-P}\right).  
\end{equation}
By subtracting (\ref{eq:PCV_marg}) in Proposition \ref{prop3} from (\ref{eq:SCCVproof2}) and regarding each term in the summation, we have
\begin{equation*}
\begin{aligned}
S_{CCV}(y_{1:n}; P) &= \frac{1}{{n \choose P}} \sum_{t=1}^{n \choose P} \left\{\log p_{\mathcal{M}}\left(\tilde{y}^{(t)}_{1:P},y^{(t)}_{1:n-P}\right)   - \log p_{\mathcal{M}}\left(y^{(t)}_{1:n-P}\right)   \right\} \\ &= \frac{1}{{n \choose P}} \sum_{t=1}^{n \choose P} \log p_{\mathcal{M}}\left(\tilde{y}^{(t)}_{1:P} \mid y^{(t)}_{1:n-P}\right)
\end{aligned}
\end{equation*}
\hphantom{0}
\end{proof}
\vspace{-10mm}

\subsection{Computing $S_{CCV}$}\label{compute_supp}
We  note that $\hat{S}_{CCV}$ in (\ref{eq:MCCCV}) is a biased estimate, and \cite{Rischard2018} provides unbiased estimators of $\log p_\mathcal{M}(\tilde{y}_{1:P} \mid y_{1:n-P})$ directly through unbiased Markov chain Monte Carlo and path sampling methods.

The arithmetic averaging over training/test splits $\hat{S}_{CCV}$ may also be inherently unstable, as demonstrated by the following example. Suppose that $y$ is a binary random variable which takes on either $0$ or $1$ with equal probability, and we are attempting to estimate $S_{CCV}(y_{1:n}; n/2)$. For large $n$, it is likely that approximately half of the values in $y_{1:n}$ are equal to 0 and the other half to 1. There will thus exist a permutation of the sequence $y_{1:n}$ such that almost all the first $n/2$ values are equal to 0, with the remaining almost all equal to 1. The model will then be certain that $y=0$ after observing the training set, and score the remaining $n/2$ points very poorly, giving a large negative log posterior predictive. This suggests that an arithmetic average may be unstable; the median or robust trimmed mean over permutations may be stabler alternatives.

The form in  (\ref{eq:CCV_marg}) relies on the conditional coherency of Bayesian updating and  scoring. Without this, $S_{CCV}$ still exists as defined in (\ref{eq:CCV_supp}), and can be directly estimated for example through
\begin{equation*}
\hat{S}_{CCV}(y_{1:n};P) = \frac{P}{T} \sum_{t=1}^T \frac{1}{p^{(t)}}  \sum_{j=1}^{p^{(t)}} s\left(\tilde{y}_{j}^{(t)} \mid y^{(t)}_{1:n-p^{(t)}}\right) 
\end{equation*}
where $p^{(t)} \sim \mathcal{U}\{1, P\}$ and the training set $y^{(t)}_{1:n-p^{(t)}}$ is sampled uniformly at random conditioned on $p^{(t)}$. This facilities alternative choices for the belief updating model and $s\left(\tilde{y} \mid y\right)$.

\subsection{Visualization of cumulative cross-validation}\label{vis_supp}

\begin{figure}[!h]
 \includegraphics[width=\textwidth]{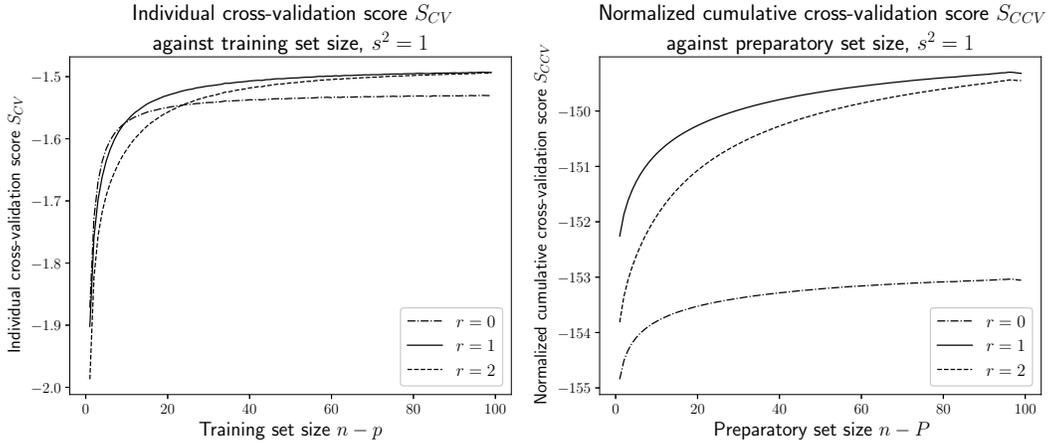}
\caption{Leave-$p$-out cross-validation score ${S}_{CV}(y_{1:n};p)$ against $n-p$ (left) and normalized cumulative cross-validation score ${S}_{CCV}(y_{1:n}; P)  \times n/P$ against $n-P$ (right) for $s^2 = 1$ and $p,P \in \{1,\ldots, 99\}$ in the polynomial regression example; the maximum standard error is 0.001 for ${S}_{CV}$  and 0.005 for $\hat{S}_{CCV}$.}
\label{fig:prep}
\end{figure}
\newpage
A visualization of the effects of the training/preparatory data size is shown in Figure \ref{fig:prep} for $s^2 = 1$ in the polynomial regression example. We omit $S_{CV}(y_{1:n};n)$ and $S_{CCV}(y_{1:n};n)$ for clarity of the plot, as both are significantly more negative than the other values. On the left we see that the individual cross-validation term $S_{CV}(y_{1:n};p)$ prefers the simplest $r=0$ model when the training set is very small as over-fitting is penalized, but as $n-p$ increases, the true $r=1$ model overtakes it. The $r=2$ model eventually overtakes the $r=0$ model too, and we see the discrepancy between $r=2$ and $r=1$ decrease as over-fitting is penalized less and less. This latter effect is demonstrative of how leave-one-out cross-validation under-penalizes complex models as argued in \citet{Shao1993}, and why a value of $P > 1$ should be preferred. On the right, we observe a similar effect for the cumulative cross-validation score $S_{CCV}$, but the discrepancy between $r=2$ and $r=1$ remains more noticeable for moderate $n-P$ as a cumulative sum of $S_{CV}$ terms is being taken.

\subsection{Illustration for the probit model}\label{probit_supp}
To demonstrate the cumulative cross-validation score in an intractable example, we carry out model selection in the Pima Indian benchmark model with a probit model. We observe binary random variables $y_{1:n}$ with associated $r$-dimensional covariates $x_{1:n}$, and the probit model is defined as
\begin{equation*}
f_\theta(y\mid x) =\left\{\Phi\left(\theta^\T \tilde{x}\right) \right\}^y \left\{1-\Phi\left(\theta^\T \tilde{x}\right) \right\}^{1-y}
\end{equation*}
where $\Phi$ is the standard normal cumulative distribution function and $\tilde{x} = \begin{bmatrix}1 & x^\T  \end{bmatrix}^\T$. As suggested in \citet{Marin2010}, we elicit a g-prior $\pi(\theta) = \mathcal{N}\left\{\theta; 0_{r+1}, g(X^\T X)^{-1}\right\} $ where $0_{r+1}$ is a $r+1$ vector of $0$s and $X$ is the $n$ by $r+1$ matrix with rows $\tilde{x}_i^\T$. 

The dataset consists of $n = 332$ data points and we consider $r=3$ covariates consisting of \texttt{glu}, \texttt{bp} and \texttt{ped}, which correspond to plasma glucose concentration from an oral glucose test, diastolic blood pressure and diabetes pedigree function respectively. We compare the full model $\mathcal{M}_0$: (\texttt{glu,bp,ped})  with  $\mathcal{M}_1$: (\texttt{glu,bp}) through $\log p_\mathcal{M} (y_{1:n})$ and $S_{CCV}(y_{1:n};P)$ to test for significance of \texttt{ped}. We standardize all covariates to have 0 mean and variance 1. We calculate $\log p_\mathcal{M} (y_{1:n})$ using importance sampling with a Gaussian proposal with $10^3$ samples. The proposal mean is set to the maximum likelihood estimate of $\theta$ and proposal covariance to the estimated covariance matrix of the maximum likelihood estimate as suggested in \cite{Marin2010}. For ${S}_{CCV}(y_{1:n};P)$, we estimate each posterior predictive in (\ref{eq:MCCCV}) with the same importance sampling scheme where we temper the proposal such that its covariance matrix is divided by $(n-P)/n$. We also use $10^3$ proposal samples and average over $T=10^5$ random train/test splits. We carry out 10 runs of each and report the mean and maximum standard error as before. 

We see in Table \ref{tbl:probit} that for $g=n$, the simpler model with \texttt{ped} omitted performs worse for both scores, and there is thus strong evidence for \texttt{ped}. However, when we set $g=10n$, we see that comparing models via the marginal likelihood suggests that \texttt{ped} is no longer significant, while the cumulative cross-validation score changes little with this increased variance of the prior. As a sanity check, we run a Gibbs sampler targeting $\pi(\theta \mid y_{1:n},x_{1:n})$ for the two prior settings within the full model $\mathcal{M}_0$, and plot the marginal posterior of $\theta_{\mathrm{ped}}$ in Figure \ref{fig:pimapost}. For reference, the posterior means of $\theta_{\mathrm{glu}},\theta_{\mathrm{bp}}$ are $0.70$ and $0.12$ respectively. The posteriors of $\theta_{\mathrm{ped}}$ are indistinguishable for the two prior settings, with a significant mean for $\theta_{\mathrm{ped}}$. This agrees well with the cumulative cross-validation score $\hat{S}_{CCV}$ which is clearly robust to vague priors.

\begin{table}[!h]
\centering
\def~{\hphantom{-}}
\caption{Log marginal likelihoods and cumulative cross-validation score for probit model}{%
\begin{small}
\begin{sc}
\begin{tabular}{l|c||c|c}
 $g$~~&Model  &$\log p_{\mathcal{M}}(y_{1:n})$&$\hat{S}_{CCV}(y_{1:n}; P) \times n/P$  \\[2pt]
 &    & &$P= 0.9n$   \\[5pt]
 \hline \hline
 $n$ 		  &(\texttt{glu,bp,ped})& {-168.93}&{-165.87}  \\
  							  &(\texttt{glu,bp})& -170.00&-167.37 \\[5pt]
$10n$   		  &(\texttt{glu,bp,ped})& -173.10&{-166.28} \\ &(\texttt{glu,bp})& {-173.05}&-167.64 \\[5pt] \hline \multicolumn{2}{r||} {Maximum standard error} & \hphantom{-000}0.004& \hphantom{-00}0.02
\end{tabular}
\end{sc}
\end{small}}
\label{tbl:probit}
\end{table}
\begin{figure}[!ht]
\begin{center}
\includegraphics[width=0.5\textwidth]{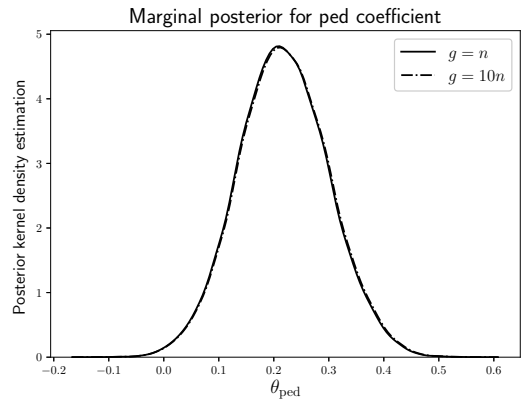}
\caption{Marginal posterior density plots for  $\theta_{\text{ped}}$  for different prior scalings $g$.}
\label{fig:pimapost}
\end{center}
\end{figure}

\end{document}